\def\llncs{0}
\documentclass{article}
\usepackage{graphicx} 
\usepackage{fullpage}
\usepackage{amsmath}

\usepackage{amsthm}
\usepackage{amssymb}
\usepackage{physics}
\usepackage{hyperref}
\usepackage{cleveref}
\usepackage{dsfont}
\usepackage{xcolor}
\usepackage{mathtools}
\usepackage{breakcites}
\usepackage{subcaption}
\usepackage[shortlabels]{enumitem}

\ifnum\llncs=0
\newtheorem{theorem}{Theorem}[section]

\newtheorem{corollary}[theorem]{Corollary}
\newtheorem{definition}[theorem]{Definition}
\newtheorem{lemma}[theorem]{Lemma} 

\newtheorem{remark}[theorem]{Remark}

\theoremstyle{definition}

\newtheorem{construction}{Construction}
\fi

\newcommand{\A}{\mathcal{A}}

\newcommand{\N}{\mathbb{N}}

\newcommand{\E}{\mathop{\mathbb{E}}}

\newcommand{\poly}{\mathsf{poly}}
\newcommand{\negl}{\mathsf{negl}}

\DeclareMathOperator{\Sim}{\mathsf{Sim}}

\DeclareMathOperator{\Gen}{\mathsf{Gen}}
\DeclareMathOperator{\Enc}{\mathsf{Enc}}
\DeclareMathOperator{\Lab}{\mathsf{Lab}}

\DeclareMathOperator{\Dec}{\mathsf{Dec}}

\DeclareMathOperator{\IND}{\mathsf{IND}}
\DeclareMathOperator{\PR}{\mathsf{PR}}
\DeclareMathOperator{\CLONE}{\mathsf{CLONE}}
\DeclareMathOperator{\IDCLONE}{\mathsf{IDCLONE}}

\newcommand{\BigPr}[2]{
\Pr\left[
\begin{array}{c}
#1
\end{array}
:
\begin{array}{c}
#2
\end{array}
\right]
}

\newcommand{\xor}{\oplus}

\newcommand{\Haar}{\mathsf{Haar}}

\newcommand{\Reg}[1]{\mathbf{#1}}

\newcommand{\secpa}{\lambda}
\newcommand{\sk}{\mathsf{sk}}
\newcommand{\ek}{\mathsf{ek}}
\newcommand{\dk}{\mathsf{dk}}
\newcommand{\ct}{\mathsf{ct}}
\newcommand{\dqre}{\mathsf{dqre}}
\newcommand{\re}{\mathsf{re}}
\newcommand{\lab}{\mathsf{lab}}
\newcommand{\nf}{\mathsf{nf}}
\newcommand{\id}{\mathsf{id}}
\newcommand{\udk}{\mathsf{udk}}
\newcommand{\uni}{\mathsf{uni}}
\newcommand{\Hyb}{\mathsf{Hyb}}

\title{A Note on Boosting Uncloneable Encryption in Microcrypt}
\author{James Bartusek\thanks{Columbia University (bartusek.james@gmail.com).} \and Eli Goldin\thanks{New York University (eli.goldin@nyu.edu). Supported by an NSF graduate research fellowship.}}

\begin{document}

\maketitle

\begin{abstract}
    In this note, we consider the setting of uncloneable encryption satisfying uncloneable indistinguishability, a form of symmetric key encryption that prevents the cloning of ciphertexts in a very strong sense. Our goal is to minimize the assumptions under which (many-time secure) uncloneable encryption is known to exist, \emph{assuming} the existence of an information-theoretic ``uncloneable bit", i.e. a one-time secure uncloneable encryption scheme for one-bit messages. We observe that if a $t\to t'$ uncloneable bit exists, then the following implications hold.
    \begin{enumerate}
        \item If many-time secure symmetric key encryption exists, then many-time secure $t\to t'$ uncloneable encryption for arbitrary-length messages exists. Since many-time secure uncloneable encryption implies many-time secure symmetric key encryption, this result is tight.
        \item If pseudorandom unitaries exist, then many-time secure $t\to t'$ uncloneable encryption for arbitrary-length messages with \emph{identical copy} security exists.
    \end{enumerate}
    These results together show that many-time secure uncloneable encryption may follow from concrete assumptions in ``microcrypt", the world of unstructured quantum cryptography that plausibly exists even if $P=NP$.
\end{abstract}
\section{Introduction}

Uncloneable encryption (UE) is a fundamental primitive in quantum cryptography that leverages both the uncloneability of quantum information and the privacy of semantically-secure encryption to yield exciting new applications. Originally introduced by \cite{BL20}, UE has since been studied in  several works (e.g. \cite{10.1007/978-3-030-90459-3_11,AKLLZ22,AKL23,Botteron:2024orj,bhattacharyya2026uncloneable,poremba_et_al:LIPIcs.ITCS.2026.109,EPRINT:CGKNY25,EPRINT:AnaGol25,bhattacharyya2026uncloneablebitexists,bartusek2026unclonableencryptionhaarrandom}) proposing constructions that satisfy different security notions under various assumptions and/or idealized models.

In this note, we focus our attention on secret-key UE with standard semantic security, often referred to as \emph{uncloneable indistinguishability}. Analogously to plain secret-key encryption (SKE), UE with \emph{one-time security} plausibly exists information-theoretically. This question has a long history of study, and has come to be known as the question of whether an ``uncloneable bit'' exists. In fact, a recent preprint \cite{bhattacharyya2026uncloneablebitexists} has claimed to solve the problem with an approach based on a random 2-design.


On the other hand, if we require \emph{many-time} (reusable) security, where a single secret key can be used to encrypt several messages,\footnote{In particular, more bits than the size of the secret key.} then UE provably requires cryptographic assumptions. A recent work \cite{bartusek2026unclonableencryptionhaarrandom} has established that reusable UE lives in ``microcrypt'', a world of very unstructured cryptography that plausibly remains secure even in the event that P = NP. Importantly, this result does not rely on the existence of an information-theoretic uncloneable bit. However, security is argued in an idealized model, namely the Haar random oracle model, and does not say anything about the existence of (reusable) UE from \emph{concrete} microcrypt assumptions.

In this note, we seek to minimize the concrete assumptions under which reusable UE exists, \emph{assuming} the existence of an information-theoretic uncloneable bit. While some previous works have studied similar questions, e.g. \cite{10.1007/978-3-030-90459-3_11,TCC:HKNY24}, we systematize and generalize existing results, establishing a clearer picture of reusable UE and the assumptions on which it relies.

\subsection{Results}

First, we show that if an uncloneable bit exists, then reusable UE is \emph{equivalent} to reusable (plain) SKE. In fact, we consider a generalization of the typical 1-to-2 UE security game to any $t < t'$, and show the equivalence between reusable $t \to t'$ UE and reusable SKE, assuming the existence of a $t \to t'$ uncloneable bit. 

\begin{theorem}[Informal]\label{thm:from-reusable-SKE}
    For any $t' > t$, assuming the existence of a $t \to t'$ uncloneable bit and reusable SKE, there exists reusable $t \to t'$ UE. 
\end{theorem}

Next, we consider variants of the definition of uncloneable encryption. In particular, the flavor of UE that results from the above theorem statement comes with the following caveats. 

\begin{itemize}
    \item The encryption and decryption key are different. That is, one needs ``more'' information to encrypt than decrypt, and the cloning game adversaries are only given what is necessary to decrypt. A stronger notion of UE is one in which the encryption and decryption key are identical, which we refer to as ``normal form'' UE.
    \item The $t$ challenge ciphertexts given to the cloning adversary are \emph{mixed states}, sampled independently as encryptions of the same message $m$. One could alternatively consider a notion where a \emph{single} encryption of $m$ is sampled and the adversary is given $t$ identical copies of this state, which we refer to as UE with ``identical copy'' security.
\end{itemize}

We next show how to obtain normal form UE as well as UE with identical copy security, though we do not know how to do so from the generic assumption of reusable SKE. Instead, we assume the existence of \emph{pseudorandom unitaries}, which is nevertheless a standard microcrypt assumption.

\begin{theorem}[Informal]\label{thm:from-PRU}
    For any $t' > t$, assuming the existence of a $t \to t'$ uncloneable bit and pseudorandom unitaries, there exists reusable \emph{normal form} $t \to t'$ UE with \emph{identical copy security}.
\end{theorem}

\subsection{Techniques}

These results follow from two main compilers, which we first introduce below and then describe in some more detail. While each has appeared in some form in the literature previously, we adapt and generalize the techniques in order to establish the theorems stated above.

\paragraph{Plaintext expansion via DQRE.} \cite{TCC:HKNY24} introduced a compiler that bootstraps one-time UE for bits to one-time UE for any length of message, assuming a type of quantum garbling scheme called decomposable quantum randomized encoding (DQRE). In this work, we propose a variant of the compiler that works in the more general $t \to t'$ setting, and observe that the notion of DQRE required to instantiate the compiler exists from \emph{any} reusable SKE scheme. We further establish that our modified compiler directly achieves \textit{reusable} $t \to t'$ security, yielding our first main theorem (\Cref{thm:from-reusable-SKE}).

The most straightforward application of this compiler yields a scheme that is \emph{not} normal form. We show how to further tweak the compiler to obtain a \emph{normal form} reusable $t \to t'$ UE scheme, additionally assuming the existence of reusable SKE with \emph{pseudorandom ciphertexts}, which follows from pseudorandom unitaries (PRUs), or even pseudorandom function-like states \cite{C:AnaQiaYue22}. These results are presented in \Cref{sec:plaintext-expansion}.


\paragraph{Identical copy security via the purification channel.} Finally, we consider recent results establishing the feasibility of \emph{identical copy security} \cite{EPRINT:AnaGol25,EPRINT:CGKNY25} that are based on the existence of one-way functions. Building on recent results in quantum learning theory~\cite{tang2026conjugatequerieshelp,pelecanos2025mixedstatetomographyreduces,girardi2026randompurificationchannelsimple}, we show that the assumption underlying these approaches can be weakened to the existence of PRUs. This allows us to obtain reusable UE with identical copy security (that is still normal form) assuming an uncloneable bit and PRUs, establishing our second main theorem (\Cref{thm:from-PRU}). This is detailed in \Cref{sec:purification}.

\subsubsection{Sketch of plaintext expansion}

As mentioned above, \cite{TCC:HKNY24} shows how to construct $1\to 2$ encryption for $\secpa$-bit messages from a $1\to 2$ uncloneable bit using a tool known as a decomposable quantum randomized encoding (DQRE). Essentially, DQRE is a procedure that takes a circuit $C$ and (potentially quantum) inputs $x_1,\dots,x_{\ell}$ and produces labels $\widehat{C},\widehat{x}_1,\dots,\widehat{x}_\ell$ such that given $\widehat{C},\widehat{x}_1,\dots,\widehat{x}_\ell$, it is possible to recover $C(x)$. Furthermore, $\widehat{C},\widehat{x}_1,\dots,\widehat{x}_\ell$ together should reveal nothing about $C$ or $x$ besides $C(x)$ (and the size of the circuit $C$). To satisfy decomposability, we further require that each of the $\widehat{x}_i$ should depend only on $x_i$ and not on any of the other (qu)bits of $x$.

To build a (one-time) $1\to 2$ scheme for long messages using DQRE and a $1\to 2$ uncloneable bit,~\cite{TCC:HKNY24} gives the following construction. The encryption key (and also the corresponding decryption key) will be a one-time secure uncloneable encryption key $\udk$ as well as $\secpa$ one-time pads $\{R_i\}_{i \in [\secpa]}$. To encrypt a message $m$, first consider the circuit $C[m]$ that outputs $m$ on all inputs, and compute the decomposable randomized encoding $\widehat{C[m]}, \widehat{0}, \widehat{\udk}_1,\dots,\widehat{\udk}_\ell$, where $\widehat{0}$ is the label for a \emph{quantum} input on sufficiently many qubits set to 0. Then, for all $i\in [\secpa]$, set $c_{i,\udk_i} = R_i \xor \widehat{\udk}_i$ and $c_{i,1-\udk_i} \gets \{0,1\}^{\secpa}$. Finally, output \[\ct = \left(\widehat{C[m]}, \widehat{0}, \{\ct_{i,b}\}_{i\in [\secpa],b\in \{0,1\}}\right).\] To decrypt, recover $\widehat{\udk}_i = R_i\xor c_{i,\udk_i}$ and evaluate $\widehat{C[m]},\widehat{0},\widehat{\udk}_1,\dots,\widehat{\udk}_\ell$ to obtain $m$.

To see why this is secure, consider a pair of challenge messages $m_0,m_1$. The challenge ciphertext is then  \[\widehat{C[m]},\widehat{0},\{\ct_{i,b}\}_{i\in [\secpa],b\in \{0,1\}}.\] Define $D[m_0,m_1]$ to be the circuit that on input $(\rho,\udk)$ applies $\Dec(\udk,\rho)\to b$ and outputs $m_b$. Note that $D[m_0,m_1](\Enc(m_b),\udk) = C[m_b](0,\udk)$. And so by the security of the randomized encoding, the challenge ciphertext looks like $\widehat{D[m_0,m_1]},\widehat{\Enc(b)}$ followed by $\{\ct_{i,b}\}_{i\in [\secpa],b\in \{0,1\}}$.

But this challenge ciphertext can be simulated given a challenge in the uncloneable bit game (i.e. either $\Enc(0)$ or $\Enc(1)$), and so a cloning attack on the construction gives a cloning attack against the $1\to 2$ uncloneable bit.

\paragraph{Reducing assumptions to reusable SKE} We first observe that DQRE can be constructed generically from reusable SKE. This follows by combining~\cite{C:AnaQiaYue22}, which shows that reusable SKE can be used to build decomposable randomized encodings for \emph{classical} circuits, and~\cite{STOC:BraYue22}, which constructs DQRE from decomposable classical randomized encodings. Although this observation is folklore, as far as we are aware it has not been written explicitly in the literature. This gives the following immediate corollary.

\begin{corollary}
    Assuming the existence of a $1\to 2$ uncloneable bit and reusable SKE, there exists one-time secure $1\to 2$ UE for any message length.
\end{corollary}

\paragraph{Lifting to $t\to t'$ security}

Unfortunately, this construction does not work in the $t\to t'$ setting for $t > 1$. This is due to the fact that the encryption key contains one-time pads, which cannot be reused. We resolve this with a natural approach: replace the one-time pad with a reusable secret key encryption scheme. Since secret key encryptions are possibly distinguishable from random strings, we also need to replace the random values $c_{i,1-\udk_i}$ with encryptions of $0$ under some fixed key.

Our resulting scheme thus involves a one-time encryption key $\udk$ as well as $2\secpa$ additional reusable encryption keys $\{\ek_{i,b}\}_{i\in [\ell],b\in \{0,1\}}$. The decryption key consists of $\udk$ as well as the decryption keys indexed by the bits of $\udk$: $\dk_{1,\udk_1},\dots,\dk_{\ell,\udk_\ell}$. To encrypt, compute the DQRE $\widehat{C[m]},\widehat{0},\widehat{\udk}_1,\dots,\widehat{\udk}_\ell$. Then, for all $i\in [\secpa]$, set $c_{i,\udk_i} = \Enc(\ek_{i,\udk_i}, \widehat{\udk}_i)$ and $c_{i,1-\udk_i} = \Enc(\ek_{i,1-\udk_i},0)$. Finally, output \[\ct = \left(\widehat{C[m]}, \widehat{0}, \{\ct_{i,b}\}_{i\in [\secpa],b\in \{0,1\}}\right).\]

As before, decryption follows by recovering $\widehat{\udk}_i = \Dec(\dk_{i,\udk_i}, c_{i,\udk_i})$ and evaluating $\widehat{C[m]},\widehat{0},\widehat{\udk}_1,\dots,\widehat{\udk}_\ell$. Security also follows from a similar argument.

It turns out that we also obtain many-time security for free using this compiler. In particular, an adversary for the uncloneable bit game simulating an adversary for this construction can simulate encryption queries without making use of the decryption key $\udk$. This can be done by setting for all $i\in [\ell],b\in \{0,1\}$, $\ct_{i,b} = \Enc(\ek_{i,b}, \widehat{b})$. When $b= \udk_i$, this is exactly of the right form, and when $b\neq \udk_i$, $\dk_{i,b}$ is never revealed and so this is indistinguishable from the honest ciphertext produced by the construction.

\begin{theorem}
    For any $t'>t$, assuming the existence of a $t\to t'$ uncloneable bit and reusable SKE, there exists reusable $t\to t'$ UE.
\end{theorem} 

\paragraph{A normal form construction.} However, this protocol is not of normal form. Even if the reusable encryption scheme is of normal form and has $\ek_{i,b}=\dk_{i,b} \eqqcolon \sk_{i,b}$, we still are not able to reveal $\sk_{i,1-\udk_i}$ to the adversary. We would thus like encryption to not make use of $\sk_{i,1-\udk_i}$ at all. We can do this by ensuring that the reusable encryption scheme used is \textit{pseudorandom}, i.e. honest encryptions should be indistinguishable from the maximally mixed state. In this case, the encryption protocol can sample a random string instead of computing $\Enc(\sk_{i,1-\udk_i},0)$. Setting $\ek=\dk=(\udk,\sk_{1,\udk_1},\dots,\sk_{\ell,\udk_\ell})$ gives a normal form reusable $t\to t'$ secure uncloneable encryption scheme from a $t\to t'$ uncloneable bit along with normal form, pseudorandom, reusable SKE.

\begin{theorem}\label{thm:intropseudorandom}
    For any $t'>t$, assuming the existence of a $t \to t'$ uncloneable bit and normal form, pseudorandom, reusable SKE, there exists normal form, reusable $t\to t'$ UE.
\end{theorem}

\subsubsection{Sketch of identical copy security}

Recent work has shown how to construct identical copy secure uncloneable encryption schemes from regular uncloneable encryption along with one-way functions~\cite{EPRINT:AnaGol25,EPRINT:CGKNY25}. A key result of these works is encapsulated in the following lemma.
\begin{lemma}[Theorem 1.1 from~\cite{EPRINT:AnaGol25}]\label{lem:simlem}
    If one-way functions exist, then there exists a quantum channel $\Sim_t$ such that for all efficiently samplable mixed states $\sigma$ there exists a family of efficiently samplable purifications $\{\ket{\phi_k}\}$ such that
    \begin{enumerate}
        \item $t$ copies of a random sample $\ket{\phi_k}$ can be efficiently approximated using $t$ i.i.d. copies of $\sigma$. That is, $$\Sim_t(\sigma^{\otimes t}) \approx \E_k\left[\ketbra{\phi_k}^{\otimes t}\right].$$
        \item $\{\ket{\phi_k}\}$ is a purification of $\sigma$. That is,
        $$\sigma \approx \E_k\left[\Tr_{\Reg{B}}(\ketbra{\phi_k}_{\Reg{A}\Reg{B}})\right].$$
    \end{enumerate}
    where here $\approx$ refers to computational indistinguishability.
\end{lemma}
This lemma intuitively means that for any cryptographic primitive with mixed state output, we can replace $t$ copies of mixed state output with $t$ copies of a random sample of the purification ensemble $\{\ket{\phi_k}\}$. This gives a generic approach for lifting ``many copy security" to ``identical copy security". In this note, we reduce the requirements for this lemma from one-way functions to pseudorandom unitaries.

Similar techniques have simultaneously appeared in the quantum learning literature~\cite{tang2026conjugatequerieshelp,pelecanos2025mixedstatetomographyreduces,girardi2026randompurificationchannelsimple}. While the simulator for these works is a bit more complex, the resulting purification is significantly simpler. The ``purification channel" stemming from the learning theory literature produces truly random purifications, not just random samples from some ensemble. 

In detail, these works show that there exists an efficient quantum channel $\Sim_t$ such that for any mixed state $\sigma_{\Reg{A}} = \Tr_{\Reg{B}}(\ketbra{\phi}_{\Reg{A}\Reg{B}}$,
$$\Sim_t(\sigma^{\otimes t}) = \E_{U\gets \Haar(M)}\left[\left((I_{\Reg{A}}\otimes U_{\Reg{B}})\ketbra{\phi}_{\Reg{A}\Reg{B}} (I_\Reg{A}\otimes U_\Reg{B}^\dagger)\right)^{\otimes t}\right].$$

Replacing $U$ with a pseudorandom unitary $U_k$ establishes that $\Sim_t(\sigma^{\otimes t})$ is indistinguishable from $U_k^{\otimes t}$ applied to the $\Reg{B}$ registers of $\ket{\phi}_{\Reg{A}\Reg{B}}^{\otimes t}$. That is,~\Cref{lem:simlem} holds under the assumption that pseudorandom unitaries exist. This immediately gives a new version of many of the transformations in~\cite{EPRINT:AnaGol25,EPRINT:CGKNY25} starting from the assumption that pseudorandom unitaries exist.

\begin{corollary}
    Suppose pseudorandom unitaries exist. Then
    \begin{enumerate}
        \item If public key quantum money exists, then multi-copy secure mini-schemes exist.
        \item For any class of functionalities $\mathcal{F}$, if there exists a copy-protection scheme for $\mathcal{F}$ satisfying i.i.d.-copy security, then there exists a copy-protection scheme for $\mathcal{F}$ satisfying identical-copy security.
        \item If there exists $t\to t'$ search-secure uncloneable encryption, then there exists $t\to t'$ identical-challenge search-secure uncloneable encryption.
    \end{enumerate}
\end{corollary}

Although the above-mentioned papers do not explicitly consider indistinguishable secure UE, we observe that the techniques do directly apply to this setting. In detail, the identical copy secure encryption algorithm will take in as its randomness a pseudorandom unitary key $k$. It will then set $\ket{\phi}_{\Reg{A}\Reg{B}}$ to be the purification of $\Enc(\ek,m)$, and output $(I_{\Reg{A}}\otimes (U_k)_{\Reg{B}}) \ket{\phi}_{\Reg{A}\Reg{B}}$. Now, challenge ciphertexts will be indistinguishable from $\Sim_t(\Enc(\ek,m)^{\otimes t})$, and so identical copy security follows from $t\to t'$ security of the underlying encryption scheme.

Note that if the underlying $t\to t'$ UE scheme is of normal form, so is the resulting identical copy secure protocol. Recalling that pseudorandom unitaries can be used to construct normal form, pseudorandom, reusable SKE~\cite{C:AnaQiaYue22}, we combine this technique with \Cref{thm:intropseudorandom} to establish the following.

\begin{corollary}\label{cor:introident}
    For any $t'>t$, assuming the existence of pseudorandom unitaries and an uncloneable $t\to t'$ bit, there exists reusable, normal form $t\to t'$ UE with identical copy security.
\end{corollary}

\section{Definitions}

Note that for all definitions, we will often omit the security parameter $\secpa$ when it is clear from context.
    
\begin{definition}
    Let $\ell_{\ek}(\secpa),\ell_{\dk}(\secpa),\ell_{m}(\secpa),\ell_{\ct}(\secpa):\N\to\N$. An encryption scheme with keys of length $\ell_{\ek},\ell_{\dk}$, messages of length $\ell_{m}$, and ciphertexts of length $\ell_{\ct}$ is a tuple of uniform quantum algorithms $\Gen,\Enc,\Dec$ with the following syntax.
    \begin{itemize}
        \item $\Gen(1^\secpa)\to \ek,\dk$ takes as input a unary security parameter $\secpa$ and outputs an encryption key $\ek\in \{0,1\}^{\ell_{\ek}(\secpa)}$ and a decryption key $\dk\in \{0,1\}^{\ell_{\dk}(\secpa)}$.
        \item $\Enc(1^\secpa,\ek,m) \to \rho_{k,m}$ takes as input a security parameter $\secpa$, an encryption key $\ek \in \{0,1\}^{\ell_{\ek}(\secpa)}$, and a message $m\in \{0,1\}^{\ell_m(\secpa)}$, and outputs a mixed state $\rho_{\ek,m} \in \mathcal{H}(\{0,1\}^{\ell_{\ct}(\secpa)})$.
        \item $\Dec(1^\secpa, \dk, \rho) \to m$ takes as input a security parameter $\secpa$, a decryption key $\dk \in \{0,1\}^{\ell_{\dk}(\secpa)}$, and a quantum ciphertext $\rho \in \mathcal{H}(\{0,1\}^{\ell_{\ct}(\secpa)})$, and outputs a decrypted message $m \in \{0,1\}^{\ell_m(\secpa)}$.
    \end{itemize}
\end{definition}

\begin{definition}
    We say that $(\Gen,\Enc,\Dec)$ is correct if for all $\secpa\in \N, m\in \{0,1\}^{\ell_m(\secpa)}$,
        $$\BigPr{\Dec(1^\secpa, \dk, \rho_{\ek,m}) = m}{\Gen(1^\secpa)\to (\ek,\dk)\\\Enc(1^\secpa, \ek, m)\to \rho_{\ek,m}} = 1.$$
\end{definition}

\begin{remark}
    One could alternatively define correctness to only hold with probability $1-\negl(\secpa)$. However, it is possible to transform any symmetric encryption scheme with imperfect correctness into one with perfect correctness without compromising security. In particular, $\Enc$ can check if the produced ciphertext state decrypts to $m$, and if it does not $\Enc$ can simply output the message in the clear~\cite{TCC:HKNY24}.
\end{remark}

\begin{definition}[Symmetric key security]
    Let $(\Gen,\Enc,\Dec)$ be a symmetric encryption scheme that satisfies correctness. For any oracle algorithm $C^{(\cdot)}$, define $
    \IND^\secpa_{C}(\Gen,\Enc,\Dec)$ to be the following game.
    \begin{enumerate}
        \item Run $\Gen(1^\secpa)\to (\ek,\dk)$
        \item Run $C^{\Enc(1^\secpa,\ek,\cdot)}\to (m_0,m_1,\textbf{st})$
        \item Sample $b\gets \{0,1\}$
        \item Set $\rho_b = \Enc(1^\secpa,\ek,m_b)$
        \item Run $C^{\Enc(1^\secpa,\ek,\cdot)}(\rho_b,\textbf{st}) \to b'$
        \item Output $1$ if and only if $b=b'$
    \end{enumerate}

    We say that $(\Gen,\Enc,\Dec)$ satisfies reusable IND-CPA security if for any (non-uniform) QPT oracle algorithm $C^{(\cdot)}$,
    $$\Pr\left[\IND^\secpa_{C}(\Gen,\Enc,\Dec)\to 1\right] \leq \frac{1}{2}+\negl(\secpa).$$
\end{definition}

\begin{definition}[Pseudorandom encryption]
    Let $(\Gen,\Enc,\Dec)$ be a symmetric encryption scheme that satisfies correctness with ciphertexts of length $\ell_{\ct}$. For any oracle algorithm $C^{(\cdot)}$, define $\PR^\secpa_{C}(\Gen,\Enc,\Dec)$ to be the following game.
    \begin{enumerate}
        \item Run $\Gen(1^\secpa)\to (\ek,\dk)$
        \item Run $C^{\Enc(1^\secpa,\ek,\cdot)}\to (m,\textbf{st})$
        \item Sample $b\gets \{0,1\}$
        \item If $b = 0$, set $\rho = \Enc(1^\secpa,\ek,m)$
        \item If $b=1$, set $\rho = I/2^{\ell_\ct}$
        \item Run $C^{\Enc(1^\secpa,\ek,\cdot)}(\rho,\textbf{st}) \to b'$
        \item Output $1$ if and only if $b=b'$
    \end{enumerate}

    We say that $(\Gen,\Enc,\Dec)$ is pseudorandom if for any (non-uniform) QPT oracle algorithm $C^{(\cdot)}$,
    $$\Pr\left[\PR^\secpa_{C}(\Gen,\Enc,\Dec)\to 1\right] \leq \frac{1}{2}+\negl(\secpa).$$
\end{definition}

\begin{definition}[Uncloneable security]
    Let $t(\secpa),t'(\secpa):\N \to \N$ and let $(\Gen,\Enc,\Dec)$ be a symmetric encryption scheme that satisfies correctness. For any algorithms $A_1,\dots,A_{t'}$ and oracle algorithm $C^{(\cdot)}$, define $\CLONE^{t\to t',\secpa}_{A_1,\dots,A_{t'},C}(\Gen,\Enc,\Dec)$ to be the following game
    \begin{enumerate}
        \item Run $\Gen(1^\secpa)\to (\ek,\dk)$
        \item Run $C^{\Enc(1^\secpa, \ek,\cdot)} \to (m_0,m_1,\textbf{st})$
        \item Sample $b\gets \{0,1\}$
        \item Set $\rho_b = \Enc(1^\secpa,\ek, m_b)$
        \item Run $C^{\Enc(1^\secpa, \ek,\cdot)}(\rho_b^{\otimes t}, \textbf{st}) \to \sigma_{\Reg{A}_1\cdots \Reg{A}_{t'}}$
        \item For each $i \in [t']$, run $A_i(1^\secpa,\dk,\sigma_{\Reg{A}_i})\to b_i$
        \item Output $1$ if and only if $b=b_1=\dots=b_{t'}$
    \end{enumerate}
    We say that $(\Gen,\Enc,\Dec)$ satisfies reusable $t\to t'$ security if for any (non-uniform) QPT $A_1,\dots,A_{t'}$ and (non-uniform) QPT oracle algorithm $C^{(\cdot)}$,
    $$\Pr[\CLONE^{t\to t',\secpa}_{A_1,\dots,A_{t'},C}(\Gen,\Enc,\Dec) \to 1] \leq \frac{1}{2}+\negl(\secpa).$$

    We say that $(\Gen,\Enc,\Dec)$ satisfies \emph{one-time} $t\to t'$ security if for any (non-uniform) QPT $A_1,\dots,A_{t'}$ and (non-uniform) QPT algorithm $C$ making no oracle queries,\footnote{One could alternatively demand security against \emph{unbounded} adversaries, though we note that security against QPT adversaries suffices for our results.}
    $$\Pr[\CLONE^{t\to t',\secpa}_{A_1,\dots,A_{t'},C}(\Gen,\Enc,\Dec) \to 1] \leq \frac{1}{2} + \negl(\secpa).$$

    We will call a symmetric encryption scheme for one-bit messages satisfying one-time $t\to t'$ security a ``$t\to t'$ uncloneable bit".
\end{definition}

\begin{definition}[\cite{bartusek2026unclonableencryptionhaarrandom}]
    We will call a symmetric encryption scheme $(\Gen,\Enc,\Dec)$ \emph{pure} if there exists a function $\ell_r(\secpa): \mathbb{N} \to \mathbb{N}$ such that the encryption algorithm $\Enc$ is of the following form.
    \begin{enumerate}
        \item Take as input a key $\ek$ and a message $m$.
        \item Sample $r\gets \{0,1\}^{\ell_r(\secpa)}$.
        \item Output $E_{\ek} \ket{m,r,0}$ where $E_{\ek}$ is some efficient unitary.
    \end{enumerate}
    We will call $E = \{E_{\ek}\}$ the purification of $\Enc$, and for $r\in \{0,1\}^{\ell_r}$, we will denote $\Enc(m;r) = E_\ek \ket{m,r,0}$.
\end{definition}

\begin{definition}
    We say that a symmetric encryption scheme $(\Gen,\Enc,\Dec)$ is of \emph{normal form} if 
    $$\Pr_{(\ek,\dk) \gets \Gen(1^\secpa)}[\ek = \dk] = 1.$$
    In this case, we will sometimes write $sk\coloneqq ek=dk$.
\end{definition}

\begin{definition}[Identical copy security]
    Let $t(\secpa),t'(\secpa):\N \to \N$ and let $(\Gen,\Enc,\Dec)$ be a pure symmetric encryption scheme that satisfies correctness. For any algorithms $A_1,\dots,A_{t'}$ and any oracle algorithm $C^{(\cdot)}$, define $\IDCLONE^{t\to t',\secpa}_{A_1,\dots,A_{t'},C}(\Gen,\Enc,\Dec)$ to be the following game
    \begin{enumerate}
        \item Run $\Gen(1^\secpa)\to (\ek,\dk)$
        \item Run $C^{\Enc(1^\secpa, \ek,\cdot)} \to (m_0,m_1,\textbf{st})$.
        \item Sample $b\gets \{0,1\}$, $r\gets \{0,1\}^{\ell_r(\secpa)}$
        \item Set $\ket{\phi_b} = \Enc(1^\secpa, \ek, m_b; r)$
        \item Run $C^{\Enc(1^\secpa, \ek,\cdot)}(\ketbra{\phi_b}^{\otimes t}, \textbf{st}) \to \sigma_{\Reg{A}_1\cdots \Reg{A}_{t'}}$
        \item For each $i \in [t']$, run $A_i(1^\secpa,\dk,\sigma_{\Reg{A}_i})\to b_i$
        \item Output $1$ if and only if $b=b_1=\dots=b_{t'}$
    \end{enumerate}
    We say that $(\Gen,\Enc,\Dec)$ satisfies reusable identical copy $t\to t'$ security if for any (non-uniform) QPT $A_1,\dots,A_{t'}$ and (non-uniform) QPT oracle algorithm $C^{(\cdot)}$,
    $$\Pr[\IDCLONE^{t\to t',\secpa}_{A_1,\dots,A_{t'},C}(\Gen,\Enc,\Dec) \to 1] \leq \frac{1}{2}+\negl(\secpa).$$
\end{definition}

\begin{remark}
    The key difference between standard security and identical copy security is that in the standard security model, the adversary gets many copies of the \textit{mixed state} encoding the challenge message $m_b$. However, in the identical copy setting, the adversary gets many copies of the \textit{pure state} corresponding to encryptions of $m_b$ under some fixed randomness $r$.
\end{remark}
\section{Plaintext expansion from reusable SKE}\label{sec:plaintext-expansion}

We first formally define decomposable randomized encodings for both classical and quantum circuits.

\begin{definition}[Decomposable classical randomized encoding]
    A decomposable classical randomized encoding is a tuple of algorithms $(\Enc,\Lab,\Dec)$ with the following syntax.
    \begin{itemize}
        \item $\Enc(1^\secpa,C) \to \widehat{C},r$ takes in a security parameter $\secpa$ and the description of a classical circuit $C$ with $\ell$ classical input bits, and outputs a (potentially quantum) encoded circuit $\widehat{C}$ as well as randomness $r$.
        \item $\Lab_i(1^\secpa,x_i,r) \to \widehat{x}_i$ takes in a security parameter $\secpa$, an index $i$, a bit $x_i$, and randomness $r$, and outputs a classical string $\widehat{x}_i$.
        \item $\Dec(1^\secpa, \widehat{C},\widehat{x}) \to y$ takes in a security parameter $\secpa$, an encoded circuit $\widehat{C}$, and classical input labels $\widehat{x}$, and produces an output string $y$.
    \end{itemize}
    It should satisfy the following.
    \begin{itemize}
        \item \textbf{Correctness}: Let $C$ be any classical circuit taking in $\ell$ classical bits. For all classical inputs $x = (x_1,\dots,x_{\ell})$, 
        $$\BigPr{\Dec(1^\secpa,\widehat{C},(\widehat{x}_1,\dots,\widehat{x}_{\ell}))\to C(x)}{\Enc(1^\secpa, C) \to \widehat{C}, r\\
        \Lab_i(1^\secpa, x_i, r)\to \widehat{x}_i} = 1.$$
        \item \textbf{Security}: There exists a QPT algorithm $\Sim$ such that for any classical circuit $C$, classical input $x=(x_1,\dots,x_{\ell})$, and (non-uniform) QPT algorithm $\A$, define
        $$p_1=\BigPr{\A(\widehat{C},(\widehat{x}_1,\dots,\widehat{x}_{\ell}))\to 1}{\Enc(1^\secpa, C) \to \widehat{C}, r\\
        \Lab_i(1^\secpa, x_i, r)\to \widehat{x}_i},$$
        $$p_2= \Pr\left[\A(\Sim(1^\secpa, |C|, \ell, C(x))) \to 1\right],$$
        then
        $$\abs{p_1-p_2}\leq \negl(\secpa).$$
    \end{itemize} 
\end{definition}

\begin{definition}[Decomposable quantum randomized encoding~\cite{STOC:BraYue22}]
    A decomposable quantum randomized encoding (DQRE) is a tuple of algorithms $(\Enc,\Lab^{q},\Lab^{c},\Dec)$ with the following syntax.
    \begin{itemize}
        \item $\Enc(1^\secpa,C, \ell_q,\ell_c) \to \widehat{C}, r, \sigma$ takes in a security parameter $\secpa$ and the description of a quantum circuit $C$ with $\ell_q$ input qubits and $\ell_c$ input bits, and outputs a (potentially quantum) encoded circuit $\widehat{C}$, classical randomness $r$, as well as a quantum resource state $\sigma_{\Reg{R}_1,\dots,\Reg{R}_{\ell_q}}$ (potentially entangled with $\widehat{C}$).
        \item $\Lab_i^q(1^\secpa,\rho_i,r,\sigma_{\Reg{R}_i}) \to \widehat{\rho}_i$ takes in a security parameter $\secpa$, an index $i$, a single qubit state $\rho_i$, randomness $r$, and the $i$th register of the resource state, and outputs a mixed state $\widehat{\rho}_i$ (possibly entangled with the resource state).
        \item $\Lab_i^c(1^\secpa,x_i,r) \to \widehat{x}_i$ takes in a security parameter $\secpa$, an index $i$, a bit $x_i$, randomness $r$, and outputs a classical string $\widehat{x}_i$.
        \item $\Dec(1^\secpa, \widehat{C},\widehat{\rho},\widehat{x}) \to \tau$ takes in a security parameter $\secpa$, an encoded circuit $\widehat{C}$, quantum input labels $\widehat{\rho}$, and classical input labels $\widehat{x}$, and produces an output state $\tau$.
    \end{itemize}
    It should satisfy the following.
    \begin{itemize}
        \item \textbf{Correctness}: Let $C$ be any quantum circuit taking in $\ell_q$ qubits and $\ell_c$ classical bits. Let $\Reg{I}_1,\dots,\Reg{I}_{\ell_q}$ be registers over $\mathcal{H}(\{0,1\})$ and let $\Reg{E}$ be an arbitrary register. For any quantum state $\rho_{\Reg{I}_1\cdots \Reg{I}_{\ell_q},\Reg{E}}$, and classical string $x = (x_1,\dots,x_{\ell_c})$, define $\tau$ to be the state produced by the following process.
        \begin{enumerate}
            \item Run $\Enc(1^\secpa, C, \ell_q,\ell_c) \to \widehat{C}, r,\sigma$
            \item For $i\in [\ell_q]$, run $\Lab_i^q(1^\secpa, \rho_{\Reg{I}_i}, r, \sigma_{\Reg{R}_i}) \to \widehat{\rho}_i$
            \item For $i \in [\ell_c]$, run $\Lab_i^c(1^\secpa, x_i, r)\to \widehat{x}_i$
            \item Output $\Dec(1^\secpa, \widehat{C},(\widehat{\rho}_1,\dots,\widehat{\rho}_{\ell_q}),(\widehat{x}_1,\dots,\widehat{x}_{\ell_c})) \to \tau$
        \end{enumerate}
        Then 
        $$(C(\rho_{\Reg{I}_1\cdots \Reg{I}_{\ell_q}},x), \rho_{\Reg{E}}) \equiv (\tau, \rho_{\Reg{E}})$$
        are equivalent as density matrices.
        \item \textbf{Security}: There exists a QPT algorithm $\Sim$ such that for any quantum state $\rho_{\Reg{I}_1\cdots \Reg{I}_{\ell_q},\Reg{E}}$, classical string  $x = (x_1,\dots,x_{\ell_c})$, quantum circuit $C$, and (non-uniform) QPT algorithm $\A$, define
        $$p_1=\BigPr{\A(\widehat{C}, (\widehat{\rho}_1,\dots,\widehat{\rho}_{\ell_q}),(\widehat{x}_1,\dots,\widehat{x}_{\ell_c}),\rho_{\Reg{E}})\to 1}{\Enc(1^\secpa, C, \ell_q, \ell_c) \to \widehat{C}, r, \sigma\\
        \Lab_i^q(1^\secpa, \rho_{\Reg{I}_i}, r, \sigma_{\Reg{R}_i}) \to \widehat{\rho}_i\\
        \Lab_i^c(1^\secpa, x_i, r)\to \widehat{x}_i},$$
        $$p_2= \Pr\left[\A(\Sim(1^\secpa, |C|, \ell_q, \ell_c, C(\rho_{\Reg{I}_1\cdots \Reg{I}_{\ell_q}},x)), \rho_{\Reg{E}}) \to 1\right],$$
        then
        $$\abs{p_1-p_2}\leq \negl(\secpa).$$
    \end{itemize} 

    Note that for classical inputs, the label is required to be classical and independent of the resource state $\sigma$.
\end{definition}

It is easy to see that security immediately implies the following useful indistinguishability property.

\begin{theorem}\label{thm:simpdqre}
    Let $C,D$ be two quantum circuits taking as input $\ell_q$ qubits and $\ell_c$ classical bits with $|C|=|D|$. Let $\rho_{\Reg{I}_{1},\dots,\Reg{I}_{\ell_q},\Reg{E}}^C,\rho_{\Reg{I}_{1},\dots,\Reg{I}_{\ell_q},\Reg{E}}^D$ be any two quantum states, and let $x^C,x^D \in \{0,1\}^{\ell_c}$ be any two classical strings such that $(C(\rho_{\Reg{I}_1\cdots \Reg{I}_{\ell_q}}^C,x^C),\rho_{\Reg{E}}^C) \equiv (D(\rho_{\Reg{I}_1\cdots \Reg{I}_{\ell_q}}^D,x^D),\rho_{\Reg{E}}^D)$ are equivalent as mixed states. Then for any QPT adversary $\A$, define
    $$p_1=\BigPr{\A(\widehat{C}, (\widehat{\rho}_{\Reg{I}_1},\dots,\widehat{\rho}_{\Reg{I}_{\ell_q}}),(\widehat{x}_1,\dots,\widehat{x}_{\ell_c}),\rho_{\Reg{E}}^C)\to 1}{\Enc(1^\secpa, C, \ell_q, \ell_c) \to \widehat{C}, r, \sigma\\
        \Lab_i^q(1^\secpa, \rho_{\Reg{I}_i}^C, r, \sigma_{\Reg{R}_i}) \to \widehat{\rho}_i\\
        \Lab_i^c(1^\secpa, x_i^C, r)\to \widehat{x}_i},$$
    $$p_2=\BigPr{\A(\widehat{D}, (\widehat{\rho}_{\Reg{I}_1},\dots,\widehat{\rho}_{\Reg{I}_{\ell_q}}),(\widehat{x}_1,\dots,\widehat{x}_{\ell_c}),\rho_{\Reg{E}}^D)\to 1}{\Enc(1^\secpa, D, \ell_q, \ell_c) \to \widehat{D}, r, \sigma\\
        \Lab_i^q(1^\secpa, \rho_{\Reg{I}_i}^D, r, \sigma_{\Reg{R}_i}) \to \widehat{\rho}_i\\
        \Lab_i^c(1^\secpa, x_i^D, r)\to \widehat{x}_i},$$
    then
    $$\abs{p_1-p_2}\leq \negl(\secpa).$$
\end{theorem}

\begin{remark}
    Unlike~\cite{STOC:BraYue22}, we explicitly require that producing classical labels does not depend on the resource state. This is a property achieved in their original protocol of~\cite{STOC:BraYue22}. In particular, producing classical labels always begins with a measurement of the corresponding register of the resource state in the standard basis, which may as well be done in the $\Enc$ step with the result stored in the randomness $r$. Although not explicitly stated, this requirement is also necessary for the construction from~\cite{TCC:HKNY24}.
\end{remark}

We then make explicit the fact that decomposable quantum randomized encodings can be constructed from reusable IND-CPA symmetric key encryption.

\begin{theorem}[\cite{STOC:BraYue22}]\label{thm:dqre}
    If there exists a decomposable classical randomized encoding, then there exists a decomposable quantum randomized encoding.
\end{theorem}

\begin{theorem}[Section 7.3 of~\cite{C:AnaQiaYue22}]\label{thm:dcre}
    If there exists a reusable IND-CPA symmetric encryption scheme, then there exists a decomposable classical randomized encoding.\footnote{In fact, the authors only need a one-time symmetric encryption scheme for messages twice as long as the key length. This follows immediately from reusable IND-CPA security.}
\end{theorem}

\begin{corollary}\label{cor:dqre}
    If there exists a reusable IND-CPA symmetric encryption scheme, then there exists a decomposable quantum randomized encoding.\footnote{Although this corollary is folklore, as far as the authors are aware it has not been explicitly stated in the literature.}
\end{corollary}

Now, we prove that plaintext expansion and reusable security for uncloneable encryption follows from the existence of reusable IND-CPA symmetric key encryption.

\begin{theorem}\label{thm:expansion}
    Let $t(\secpa),t'(\secpa),\ell_m(\secpa)$ be polynomials. If there exists a symmetric encryption scheme with one-time $t\to t'$ security for one-bit messages and there exists a symmetric encryption scheme with reusable IND-CPA security, then there exists a symmetric encryption scheme with reusable $t\to t'$ security for $\ell_m(\secpa)$ bit messages.
\end{theorem}

\begin{proof}
We first present the construction, and then analyze its security.
\begin{construction}
    Let $(\Gen^{1},\Enc^1,\Dec^1)$ be a symmetric encryption scheme with one-time $t\to t'$ security for one-bit messages, with ciphertexts of length $\ell_{\ct}^1$ and decryption keys of length $\ell_{\dk}^1$. Let \[(\Enc^{\dqre},\Lab^{\dqre,q},\Lab^{\dqre,c},\Dec^{\dqre})\] be a decomposable quantum randomized encoding, as is guaranteed by~\Cref{cor:dqre}. Let $(\Gen^{\re},\Enc^{\re},\Dec^{\re})$ be a reusable IND-CPA secure symmetric encryption scheme. We define $(\Gen^{\ell_m},\Enc^{\ell_m},\Dec^{\ell_m})$ as follows.
    \begin{enumerate}
        \item $\Gen^{\ell_m}(1^\secpa)$: 
        \begin{enumerate}
            \item Sample $\Gen^{1}(1^\secpa) \to (\ek^{1},\dk^1)$. Define $\dk^1_{i}$ to be the $i$th bit of $\dk^1$, and define $\ell \coloneqq \ell_\dk^1$ to be the bit-length of $\dk^1$.
            \item For $i\in [\ell]$, $b\in \{0,1\}$, run $\Gen^{\re}(1^\secpa) \to (\ek_{i,b}, \dk_{i,b})$.
            \item Output $\ek = \left(\{\ek^{\re}_{i,b}\}_{i\in [\ell],b\in \{0,1\}},\dk^1\right), \dk = \left(\dk^{\re}_{1,\dk^1_1},\dots,\dk^{\re}_{\ell,\dk^1_{\ell}},\dk^1\right)$.
        \end{enumerate}
        \item $\Enc^{\ell_m}(1^\secpa,\ek,m)$: 
        \begin{enumerate}
            \item Parse $\ek = (\{\ek^{\re}_{i,b}\}_{i\in [\ell],b\in \{0,1\}},\dk^1)$.
            \item Define $C[m]$ to be the circuit that always outputs $m$, appropriately padded. It will take in $\ell_{\ct}^1$ qubits and $\ell = \ell_{\dk}^1$ classical bits. 
            \item Run $\Enc^{\dqre}(1^\secpa, C[m], \ell_{\ct}^1, \ell_{\dk}^1) \to \widehat{C},r,\sigma$.
            \item For $i \in [\ell_{\ct}^1]$, run $\Lab_i^{\dqre,q}(1^\secpa, 0, r, \sigma_{\Reg{R}_i}) \to \widehat{\rho}_i$.
            \item For $i \in [\ell_{\dk}^1]$, run $\Lab_i^{\dqre,c}(1^\secpa, \dk^1_i, r) \to \lab_{i,\dk^1_i}$. Set $\lab_{i, 1-\dk^1_i} = 0\dots 0$.
            \item For $i\in [\ell_{\dk}^1]$, $b\in \{0,1\}$ set $\ct_{i,b} = \Enc^{\re}(1^\secpa, \ek_{i,b}^{\re}, \lab_{i,b})$.
            \item Output $\ct = \left(\widehat{C},\widehat{\rho}_1,\dots,\widehat{\rho}_{\ell_{\ct}^1},\{\ct_{i,b}\}_{i\in [\ell_{\dk}^1],b\in \{0,1\}}\right)$.
        \end{enumerate}
        \item $\Dec^{\ell_m}(\dk,\ct)$:
        \begin{enumerate}
            \item Parse $\dk = \left(\dk^{\re}_{1,\dk^1_1},\dots,dk^{re}_{\ell,dk^1_{\ell}},\dk^1\right)$.
            \item Parse $\ct = \left(\widehat{C},\widehat{\rho}_1,\dots,\widehat{\rho}_{\ell_{\ct}^1},\{ct_{i,b}\}_{i\in [\ell_{\dk}^1],b\in \{0,1\}}\right)$.
            \item For each $i\in [\ell_{\dk}^1]$, run $\Dec^\re(\dk_i^{\re}, \ct_{i,\dk^1_i}) \to \lab_{i,\dk^1_i}$.
            \item Run $\Dec^{\dqre}\left(\widehat{C},(\widehat{\rho}_1,\dots,\widehat{\rho}_{\ell_{\ct}^1}),(\lab_{1,\dk^1_1},\dots,\lab_{\ell,\dk^1_\ell})\right) \to m$.
            \item Output $m$.
        \end{enumerate}
    \end{enumerate}
\end{construction}

Correctness is clear by construction. It remains to show $t\to t'$ security.

Let $A_1^{\ell_m},\dots,A_{t'}^{\ell_m},C^{\ell_m}$ be an attacker against the $t\to t'$ security of $(\Gen^{\ell_m},\Enc^{\ell_m},\Dec^{\ell_m})$. We will construct an attacker $A_1^{1},\dots,A_{t'}^1,C^1$ against the $t\to t'$ security of $(\Gen^1,\Enc^1,\Dec^1)$ simulating $A_1^{\ell_m},\dots,A_{t'}^{\ell_m},C^{\ell_m}$ as follows.
\begin{enumerate}
    \item For each $i\in [\ell_{\dk}^1]$, $b\in \{0,1\}$, $C^1$ will run $\Gen^{\re}(1^\secpa) \to (\ek_{i,b}^{\re},\dk_{i,b}^{\re})$.
    \item $C^1$ will begin by simulating $C^{\ell_m,(\cdot)} \to (m_0,m_1,\textbf{st})$ as specified next.
    \item Whenever $C^{\ell_m}$ makes a query to its oracle on input $m$, $C^1$ will do the following:
    \begin{enumerate}
        \item Run $\Enc^{\dqre}(1^\secpa,C[m],\ell_{\ct}^1,\ell_{\dk}^1)\to \widehat{C},r,\sigma$.
        \item For $i \in [\ell_{\ct}^1]$, run $\Lab_i^{\dqre,q}(1^\secpa, 0, r, \sigma_{\Reg{R}_i}) \to \widehat{\rho}_i$.
        \item For $i \in [\ell_{\dk}^1]$, $b\in \{0,1\}$, run $\Lab_i^{\dqre,c}(1^\secpa, b, r) \to \lab_{i,b}$.
        \item For $i\in [\ell_{\dk}^1]$, $b\in \{0,1\}$ set $\ct_{i,b} = \Enc^{\re}(1^\secpa, \ek_{i,b}^{\re}, \lab_{i,b})$.
        \item Output $\widehat{C},\widehat{\rho}_1,\dots,\widehat{\rho}_{\ell_{\ct}^1},\{\ct_{i,b}\}_{i\in [\ell_{\dk}^1],b\in \{0,1\}}$.
    \end{enumerate}
    \item $C^1$ will then query its own challenger with $m_0=0,m_1=1$, receiving response $\rho_1,\dots,\rho_t$. Let $(\rho_j)_{\Reg{I}_i}$ denote the $i$th qubit of $\rho_j$.
    \item For each $j \in [t]$, $C^1$ will then do the following:
    \begin{enumerate}
        \item Define a circuit $D[m_0,m_1]$ that behaves as follows:
        \begin{enumerate}
            \item On input $\rho,\dk^{1}$, run $\Dec(\dk^1,\rho)\to b$
            \item Output $m_b$
        \end{enumerate}
        \item Run $\Enc^{\dqre}(1^\secpa, D[m_0,m_1], \ell_{\ct}^1, \ell_{\dk}^1) \to \widehat{D},r,\sigma$.
        \item For $i \in [\ell_{ct}^1]$, run $\Lab_i^{\dqre,q}(1^\secpa, (\rho_j)_{\Reg{I}_i}, r, \sigma_{\Reg{R}_i}) \to \widehat{\rho}_i$
        \item For $i \in [\ell_{\dk}^1]$, $b\in \{0,1\}$, run $\Lab_i^{\dqre,c}(1^\secpa, b, r) \to \lab_{i,b}$.
        \item For $i\in [\ell_{\dk}^1]$, $b\in \{0,1\}$ set $\ct_{i,b} = \Enc^{\re}(1^\secpa, \ek_{i,b}^{\re}, \lab_{i,b})$.
        \item Set $CT_{j}\coloneqq \left(\widehat{D},\widehat{\rho}_1,\dots,\widehat{\rho}_{\ell_{\ct}^1},\{\ct_{i,b}\}_{i\in [\ell_{\dk}^1],b\in \{0,1\}}\right)$.
    \end{enumerate}
    \item $C^1$ will then send challenge response $CT_1,\dots,CT_t$ to $C^{\ell_m}(CT_1,\dots,CT_{t},\textbf{st}) \to \sigma_{\Reg{A}_1,\dots,\Reg{A}_{t'}}$, simulating queries in the same way as before.
    \item $C^1$ will then send to each $A_1^j$ the state $(\sigma_{\Reg{A}_j},\{\dk_{i,b}^{\re}\}_{i\in [\ell^1_{\dk}],b\in \{0,1\}})$.
    \item For each $j\in [t']$, on input $(\sigma_{\Reg{A}_j},\{\dk_{i,b}^{\re}\}_{i\in [\ell_{\dk}^1],b\in \{0,1\}},\dk^1)$, $A_j^1$ will run $A_j^{\ell_m}(\sigma_{\Reg{A}_j}, (\dk^{\re}_{1,\dk^1_{1}},\dots,\dk^{\re}_{\ell,\dk^1_\ell})) \to b_j$ and output $b_j$.
\end{enumerate}

It then remains to show that the advantage of $A^{\ell_m}_1,\dots,A^{\ell_m}_{t'},C^{\ell_m}$ in the uncloneable encryption game is tied to the advantage of $A_1^1,\dots,A_{t'}^1,C^1$. To do this, we consider an intermediate encryption scheme $\Hyb$ which will be the same as $\CLONE^{t\to t'}_{A_1^1,\dots,A_{t'}^1,C}(\Gen^1,\Enc^1,\Dec^1)$ except that we modify lines 3.d and 5.e. In particular, let $\dk^1$ be the challenger's decryption key produced by $\Gen^1 \to (\ek^1,\dk^1)$. For each $i\in [\ell_{\dk}^1]$, instead of setting $\ct_{i,1-\dk^1_i}=\Enc^{\re}(\ek_{i,1-\dk^1_i}^{\re},\lab_{i,1-\dk^1_i})$, we set $\ct_{i,1-\dk^1_i} = \Enc^{\re}(\ek_{i,1-\dk^1_i}^{\re},0\dots 0)$.

Note that aside from constructing these ciphertexts, the entire game is independent of $\dk^{\re}_{i,1-\dk^1_i}$. A standard hybrid argument along with IND-CPA security of $(\Gen^{\re},\Enc^{\re},\Dec^{\re})$ shows that 
$$\abs{\Pr[\Hyb\to 1] - \Pr[\CLONE^{t\to t'}_{A_1^1,\dots,A_{t'}^1,C^1}(\Gen^1,\Enc^1,\Dec^1)\to 1]} \leq \negl(\secpa).$$

But now note that the only difference between $\CLONE^{t\to t'}_{A_1^{\ell_m},\dots,A_{t'}^{\ell_m},C^{\ell_m}}(\Gen^{\ell_m},\Enc^{\ell_m},\Dec^{\ell_m})$ and $\Hyb$ is the choice of circuit and input encoded by the decomposable quantum randomized encoding in the challenge ciphertexts. In the former case, we have $\widehat{C[m_b]},\widehat{0},\widehat{\dk^1}$, and in the latter case we have $\widehat{D[m_0,m_1]},\widehat{\Enc(\ek^1,b)},\widehat{\dk^1}$. But note that with probability $1$, $D[m_0,m_1](\Enc(\ek^1,b),\dk^1) = m_b = C[m_b](0,\dk^1)$. And so by a standard hybrid argument along with security of the decomposable quantum randomized encoding (\Cref{thm:simpdqre}), we have
$$\abs{\Pr[\Hyb\to 1] - \Pr[\CLONE^{t\to t'}_{A_1^{\ell_m},\dots,A_{t'}^{\ell_m},C^{\ell_m}}(\Gen^{\ell_m},\Enc^{\ell_m},\Dec^{\ell_m}) \to 1]} \leq \negl(\secpa).$$
We critically make use of the fact that security holds against quantum side information on the input, which here models the other challenge ciphertexts.

And so since 
$$\Pr[\CLONE^{t\to t'}_{A_1^1,\dots,A_{t'}^1,C^1}(\Gen^1,\Enc^1,\Dec^1)\to 1] \leq \frac{1}{2}+\negl(\secpa),$$
we have
$$\Pr[\CLONE^{t\to t'}_{A_1^{\ell_m},\dots,A_{t'}^{\ell_m},C^{\ell_m}}(\Gen^{\ell_m},\Enc^{\ell_m},\Dec^{\ell_m}) \to 1] \leq \frac{1}{2} + \negl(\secpa).$$
\end{proof}

Finally, we show that our construction can be made \emph{normal form} assuming that the underlying IND-CPA symmetric encryption scheme is both normal form and pseudorandom.

\begin{corollary}\label{cor:pseudorandom}
    Let $t(\secpa),t'(\secpa),\ell_m(\secpa)$ be polynomials. If there exists a symmetric encryption scheme with one-time $t\to t'$ security for one-bit messages and there exists a normal form, pseudorandom symmetric encryption scheme with reusable IND-CPA security, then there exists a normal form symmetric encryption scheme with reusable $t\to t'$ security for $\ell_m(\secpa)$ bit messages.
\end{corollary}

\begin{proof} We first present the construction, and then analyze its security.
    \begin{construction}
    Let $(\Gen^{1},\Enc^1,\Dec^1)$ be a symmetric encryption scheme with one-time $t\to t'$ security for one-bit messages, with ciphertexts of length $\ell_{\ct}^1$ and decryption keys of length $\ell_{\dk}^1$. Let \[(\Enc^{\dqre},\Lab^{\dqre,q},\Lab^{\dqre,c},\Dec^{\dqre})\] be a decomposable quantum randomized encoding, as is guaranteed by~\Cref{cor:dqre}. Let $(\Gen^{\re},\Enc^{\re},\Dec^{\re})$ be a pseudorandom, normal form, reusable IND-CPA secure symmetric encryption scheme with ciphertexts of length $\ell_{\ct}^{\re}$. We define $(\Gen^{\nf},\Enc^{\nf},\Dec^{\nf})$ as follows.
    \begin{enumerate}
        \item $\Gen^{\nf}(1^\secpa)$: 
        \begin{enumerate}
            \item Sample $\Gen^{1}(1^\secpa) \to (\ek^{1},\dk^1)$. Define $\dk^1_{i}$ to be the $i$th bit of $\dk^1$.
            \item For $i\in [\ell_{\dk}^1]$, run $\Gen^{\re}(1^\secpa) \to \sk^\re_{i}$.
            \item Output $\sk = \left(\sk^{\re}_{1},\dots,\sk^{\re}_{\ell_{\dk}^1},\dk^1\right)$.
        \end{enumerate}
        \item $\Enc^\nf(1^\secpa,\sk,m)$: 
        \begin{enumerate}
            \item Parse $\sk = \left(\sk^{\re}_{1},\dots,\sk^{\re}_{\ell_{\dk}^1},\dk^1\right)$.
            \item Define $C[m]$ to be the circuit which always outputs $m$, appropriately padded. It will take in $\ell_{\ct}^1$ qubits and $\ell_{\dk}^1$ classical bits. 
            \item Run $\Enc^{\dqre}(1^\secpa, C[m], \ell_{\ct}^1, \ell_{\dk}^1) \to \widehat{C},r,\sigma$.
            \item For $i \in [\ell_{\ct}^1]$, run $\Lab_i^{\dqre,q}(1^\secpa, 0, r, \sigma_{\Reg{R}_i}) \to \widehat{\rho}_i$.
            \item For $i \in [\ell_{\dk}^1]$, run $\Lab_i^{\dqre,c}(1^\secpa, \dk^1_i, r) \to \lab_{i}$.
            \item For $i\in [\ell_{\dk}^1]$, set $\ct_{i,\dk^1_i} = \Enc^{\re}(1^\secpa, \sk_{i}^{\re}, \lab_{i})$.
            \item For $i\in [\ell_{\dk}^1]$, sample $\ct_{i,1-\dk^1_i}\gets \{0,1\}^{\ell_{\ct}^{\re}}$.
            \item Output $\ct = \left(\widehat{C},\widehat{\rho}_1,\dots,\widehat{\rho}_{\ell_{\ct}^1},\{\ct_{i,b}\}_{i\in [\ell_{\dk}^1],b\in \{0,1\}}\right)$.
        \end{enumerate}
        \item $\Dec^\nf(\sk,\ct)$:
        \begin{enumerate}
            \item Parse $\sk = \left(\sk^{\re}_{1},\dots,\sk^{\re}_{\ell_{\dk}^1},\dk^1\right)$
            \item Parse $\ct = \left(\widehat{C},\widehat{\rho}_1,\dots,\widehat{\rho}_{\ell_{\ct}^1},\ct_{1,0},\dots, \ct_{\ell_{\dk}^1,0},\ct_{1,1},\dots,\ct_{\ell_{\dk}^1,1}\right)$.
            \item For each $i\in [\ell_{\dk}^1]$, run $\Dec^\re(\sk_i^{\re}, \ct_{i,\dk^1_i}) \to \lab_{i}$.
            \item Run $\Dec^{\dqre}(\widehat{C},(\widehat{\rho}_1,\dots,\widehat{\rho}_{\ell_{\ct}^1}),(\lab_{1},\dots,\lab_{\ell_{\dk}^1})) \to m$.
            \item Output $m$.
        \end{enumerate}
    \end{enumerate}

    It is again clear by construction that the protocol is correct and of normal form. It remains to show security.

    Let $A_1^\nf,\dots,A_{t'}^\nf,C^\nf$ be any adversary against $(\Gen^\nf,\Enc^\nf,\Dec^\nf)$. We observe that
    \begin{equation*}
        \begin{split}
            \Bigg | \Pr[\CLONE^{t\to t'}_{A_1^\nf,\dots,A_{t'}^\nf,C^\nf}(\Gen^\nf,\Enc^\nf,\Dec^\nf)\to 1]
            - \Pr[\CLONE^{t\to t'}_{A_1^\nf,\dots,A_{t'}^\nf,C^\nf}(\Gen^{\ell_m},\Enc^{\ell_m},\Dec^{\ell_m})\to 1]\Bigg |
            \leq \negl(\secpa),
        \end{split}
    \end{equation*}
    which follows directly from pseudorandomness of $(\Gen^{\re},\Enc^{\re},\Dec^{\re})$. In particular, $$\CLONE^{t\to t'}_{A_1^\nf,\dots,A_{t'}^\nf,C^\nf}(\Gen^\nf,\Enc^\nf,\Dec^\nf)$$ is exactly $$\CLONE^{t\to t'}_{A_1^\nf,\dots,A_{t'}^\nf,C^\nf}(\Gen^{\ell_m},\Enc^{\ell_m},\Dec^{\ell_m})$$ but with $\ct_{i,1-\dk^1_i}=\Enc^{\re}(1^\secpa,\ek_{i,1-\dk^1_i}^{\re},\lab_{i,1-\dk^1_i})$ replaced by $\ct_{i,1-\dk^1_i}\gets \{0,1\}^{\ell^{\re}_{\ct}}$. Since this is the only place in the game where $(\ek_{i,1-\dk^1_i},\dk_{i,1-\dk^1_i})$ appear, indistinguishability follows directly from pseudorandomness of $(\Gen^{\re},\Enc^{\re},\Dec^{\re})$. Security then follows by the same proof as \Cref{thm:expansion}. 
    \end{construction}
\end{proof}
\section{Identical copy security from PRUs}\label{sec:purification}

We first define the purification channel we will use to instantiate the transformations of~\cite{EPRINT:AnaGol25,EPRINT:CGKNY25}. 

\begin{theorem}[\cite{tang2026conjugatequerieshelp,girardi2026randompurificationchannelsimple}]\label{thm:purechannel}
    For all $t,N,M$, there exists a quantum channel $\Sim_t$ such that for all pure states $\ket{\phi}_{\Reg{A}\Reg{B}} \in \mathcal{H}([N])\otimes \mathcal{H}([M])$ where $\sigma_{\Reg{A}} = \Tr_{\Reg{B}}(\ket{\phi}_{\Reg{A}\Reg{B}})$,
    $$\Sim_t(\sigma_{\Reg{A}}^{\otimes t}) = \E_{U\gets \Haar(M)}\left[\left((I_{\Reg{A}}\otimes U_{\Reg{B}})\ketbra{\phi}_{\Reg{A}\Reg{B}} (I_{\Reg{A}}\otimes U_{\Reg{B}}^\dagger)\right)^{\otimes t}\right].$$
    Furthermore, there is an algorithm implementing $\Sim_t$ up to error $\epsilon$ running in time $\poly(t,\log N, \log M, \log \frac{1}{\epsilon})$.
\end{theorem}

Instantiating $U$ with a pseudorandom unitary (defined below) yields an \emph{efficiently sampleable} family of purifications that is indistinguishable from $\Sim_t(\sigma^{\otimes t})$.

\begin{definition}
    Let $\ell(\secpa)$ be any polynomial. A pseudorandom unitary acting on $\ell(\secpa)$ qubits is a family of efficiently implementable unitaries $\{U_{k}\}_{k\in \{0,1\}^{\secpa}}$ such that for all (non-uniform) QPT oracle algorithms $\A$,
    $$\abs{\Pr_{U\gets \Haar(\{0,1\}^\ell)}[\A^U\to 1] - \Pr_{k\gets \{0,1\}^\secpa}[\A^{U_k} \to 1]} \leq \negl(\secpa).$$
\end{definition}

Note that we don't require security with inverse access to $U_k$ in the above PRU definition. Now, we instantiate the transformation of~\cite{EPRINT:AnaGol25,EPRINT:CGKNY25} using pseudorandom unitaries.

\begin{corollary}
    Let $t(\secpa)<t'(\secpa)$ be any polynomials. If pseudorandom unitaries exist and there exists a symmetric encryption scheme with reusable $t\to t'$ security, then there exists a pure symmetric encryption scheme with reusable $t \to t'$ identical copy security. Furthermore, if the original uncloneable encryption scheme is of normal form, so is the identical copy secure scheme.
\end{corollary}

\begin{proof} We first present the construction, and then analyze its security.

\begin{construction}
    Let $(\Gen^{\re},\Enc^{\re},\Dec^{\re})$ be an uncloneable encryption scheme. Let $\{U_k\}_{k\in \{0,1\}^\secpa}$ be a pseudorandom unitary.

    Since $\Enc^{\re}$ is an efficient quantum channel, there exists an efficient unitary $E^{\re}$ such that $\Enc^{\re}(\ek,m)$ is exactly the following process.
    \begin{enumerate}
        \item Run $E^{\re}\ket*{\ek,m,0^{a(\secpa)}} \to \ket*{\phi_{\ek,m}}$.
        \item Treat $\ket{\phi_{\ek,m}}$ as a state over some specified registers $\Reg{A},\Reg{B}$
        \item Output $\Tr_{\Reg{B}}(\ketbra{\phi_{\ek,m}}_{\Reg{A}\Reg{B}})$.
    \end{enumerate}
    
    We construct $(\Gen^{\id},\Enc^{\id},\Dec^{\id})$ as follows.
    \begin{enumerate}
        \item $\Gen^{\id} = \Gen^{\re}$.
        \item $\Enc^{\id}(\ek,m;k) = (I_{\Reg{A}} \otimes (U_k)_{\Reg{B}})(E^{\re}\ket*{\ek,m,0^{a(\secpa)}})_{\Reg{A}\Reg{B}}$.
        \item $\Dec^{\id}(\dk,\rho_{\Reg{A}\Reg{B}})$ will output $\Dec^{\re}(\dk,\Tr_{\Reg{B}}(\rho_{\Reg{A}\Reg{B}}))$.
    \end{enumerate}
\end{construction}

To show correctness, note that 
\begin{equation*}
    \begin{split}
        \Tr_{\Reg{B}}\left((I_{\Reg{A}} \otimes (U_k)_{\Reg{B}})(E^{\re}\ket*{\ek,m,0^{a(\secpa)}})_{\Reg{A}\Reg{B}}\right) =\Tr_{\Reg{B}}\left((E^{\re}\ket*{\ek,m,0^{a(\secpa)}})_{\Reg{A}\Reg{B}}\right) = \Enc^{\re}(\ek,m),
    \end{split}
\end{equation*}
and thus, $\Dec^{\id}(\dk,\rho_{\Reg{A}\Reg{B}}) = \Dec^{\re}(\dk, \Enc^{\re}(\ek,m))=m$ with probability $1$.

It thus remains to show reusable $t\to t'$ identical copy security. Given an adversary $(A_1^{\id},\dots,A_{t'}^{\id},C^{\id})$ against $(\Gen^{\id},\Enc^{\id},\Dec^{\id})$, we can construct an adversary $(A_1^{\re},\dots,A_{t'}^{\re},C^{\re})$ against $(\Gen^{\re},\Enc^{\re},\Dec^{\re})$. $(A_1^{\re},\dots,A_{t'}^{\re},C^{\re})$ will simulate $(A_1^{\id},\dots,A_{t'}^{\id},C^{\id})$ as follows.
    \begin{enumerate}
        \item When $C^{\id}$ makes a query $m$ to its oracle, query $\Enc(m) \to \rho$. Return $\Sim_1(\rho)$ to $C^{\id}$.
        \item When $C^{\id}$ outputs a challenge $m_0,m_1$, forward $m_0,m_1$ to the challenger and receive response $\rho_b^{\otimes t}$. Return $\Sim_t(\rho_b^{\otimes t})$ to $C^{\id}$.\footnote{In full formality, $C^{id}$ will run the efficient algorithm implementing $\Sim_t$ up to error $1/2^\secpa$.}
        \item When $C^{\id}$ outputs $\sigma_{A_1,\dots,A_{t'}}$, forward $\sigma_{A_i}$ to party $A_i^{\id}$. Retrieve key $\dk$ from the challenger and forward $\dk$ to each $A_i$.
    \end{enumerate}

    Let us define an intermediate scheme $(\Gen^{\uni},\Enc^{\uni},\Dec^{\uni})$ which will essentially be $(\Gen^{\id},\Enc^{\id},\Dec^{\id})$ but with the pseudorandom unitary $U_k$ replaced with a true random unitary. In particular, $\Enc^{\uni}(\ek,m;U)$ takes as randomness the description of a Haar random unitary $U$ and outputs $(I_{\Reg{A}}\otimes U_{\Reg{B}}) (E_{\ek}^{\re}\ket{\ek,m,0^{a(\secpa)}})_{\Reg{A}\Reg{B}}$. We then define $\Gen^{\uni},\Dec^{\uni}$ to be the same as $\Gen^{\id},\Dec^{\id}$ respectively.

    \Cref{thm:purechannel} gives us that
    \begin{equation*}
        \begin{split}
            \bigg|\Pr[\CLONE^{t\to t',\lambda}_{A_1^{\re},\cdots,A_{t'}^{\re},C^{\re}}(\Gen^{\re},\Enc^{\re},\Dec^{\re}) \to 1]\\
            - \Pr[\IDCLONE^{t\to t',\lambda}_{A_1^{\id},\cdots,A_{t'}^{\id},C^{\id}}(\Gen^{\uni},\Enc^{\uni},\Dec^{\uni}) \to 1]\bigg|\\
            \leq \negl(\secpa).
        \end{split}
    \end{equation*}

    But pseudorandom unitary security gives us
    \begin{equation*}
        \begin{split}
            \bigg|\Pr[\IDCLONE^{t\to t',\lambda}_{A_1^{\id},\cdots,A_{t'}^{\id},C^{\id}}(\Gen^{\uni},\Enc^{\uni},\Dec^{\uni}) \to 1]\\
            - \Pr[\IDCLONE^{t\to t',\lambda}_{A_1^{\id},\cdots,A_{t'}^{\id},C^{\id}}(\Gen^{\id},\Enc^{\id},\Dec^{\id}) \to 1]\bigg|\\
            \leq \negl(\secpa).
        \end{split}
    \end{equation*}

    And so by security of $(\Gen^{\re},\Enc^{\re},\Dec^{\re})$, we conclude
    $$\Pr[\IDCLONE^{t\to t',\lambda}_{A_1^{\id},\cdots,A_{t'}^{\id},C^{\id}}(\Gen^{\id},\Enc^{\id},\Dec^{\id}) \to 1]\leq \frac{1}{2} + \negl(\secpa).$$
\end{proof}

Composing this result with~\Cref{cor:pseudorandom} along with the fact that pseudorandom unitaries imply reusable IND-CPA symmetric key encryption (of normal form)\cite{C:AnaQiaYue22}, we obtain the following corollary.

\begin{corollary}
    Let $t(\secpa)<t'(\secpa)$ be any polynomials. If pseudorandom unitaries exist and there exists a $t\to t'$ uncloneable bit, then there exists a pure, normal form, symmetric encryption scheme with reusable $t\to t'$ identical copy security.
\end{corollary}

\bibliographystyle{alpha}
\bibliography{refs/abbrev0,refs/crypto,refs/refs}

\end{document}